\documentclass[journal,10pt,twocolumn,twoside]{IEEEtran}
\IEEEoverridecommandlockouts
\usepackage{cite}
\usepackage{amsmath,amssymb,amsfonts}
\usepackage{algorithmic}
\usepackage{graphicx}
\usepackage{textcomp}
\usepackage{xcolor}
\usepackage{microtype}
\usepackage{booktabs}
\usepackage{tabularx}
\usepackage{amsthm}
\newtheorem{theorem}{Theorem}

\def\BibTeX{{\rm B\kern-.05em{\sc i\kern-.025em b}\kern-.08em
    T\kern-.1667em\lower.7ex\hbox{E}\kern-.125emX}}
\begin{document}

\title{Beyond Byzantine: An Organizational Consensus Algorithm for Self-Interested Agents Under Information Asymmetry}

\author{\IEEEauthorblockN{1\textsuperscript{st} Jiawei Zhang}
\IEEEauthorblockA{\textit{Department of Agricultural Economics} \\
\textit{Purdue University}\\
West Lafayette, USA \\
zhan6005@purdue.edu} \\
\and
\IEEEauthorblockN{2\textsuperscript{nd} Jianbo Liu}
\IEEEauthorblockA{
  \textit{Basic Technology Group, Intelligent Terminal Business Dept.} \\
  \textit{JD Logistics, JD.com}\\
  Beijing, China \\
  liujianbo15@jd.com} 
}

\maketitle

\begin{abstract}
Traditional distributed consensus protocols force nodes into a false binary: honest-but-faulty or actively malicious (Byzantine). Real organizational departments rarely fit such extreme categories. Instead, departmental agents act with bounded rationality, pursue localized self-interest, and operate under severe information asymmetry. We introduce the Organizational Consensus Algorithm (OCA), a mechanism design framework tailored for internal negotiation. OCA models inter-departmental conflict as a dynamic game of incomplete information, combining asset staking, exception-triggered signaling, and confidence-weighted consensus rules. Rather than enforcing instantaneous total ordering, OCA uses a retrospective penalty system anchored in delayed, verifiable outcomes to curb structural bias and eliminate redundant coordination. Evaluating OCA through a Python simulation prototype across diverse organizational scales reveals lower coordination overhead, richer informative reporting, and bounded welfare loss in noisy environments. These empirical gains remain conditional on our simulation parameters and do not, on their own, constitute a general truthful equilibrium.
\end{abstract}


\section{Introduction}

\begin{figure*}[t!]
    \centering
    \includegraphics[width=0.9\linewidth]{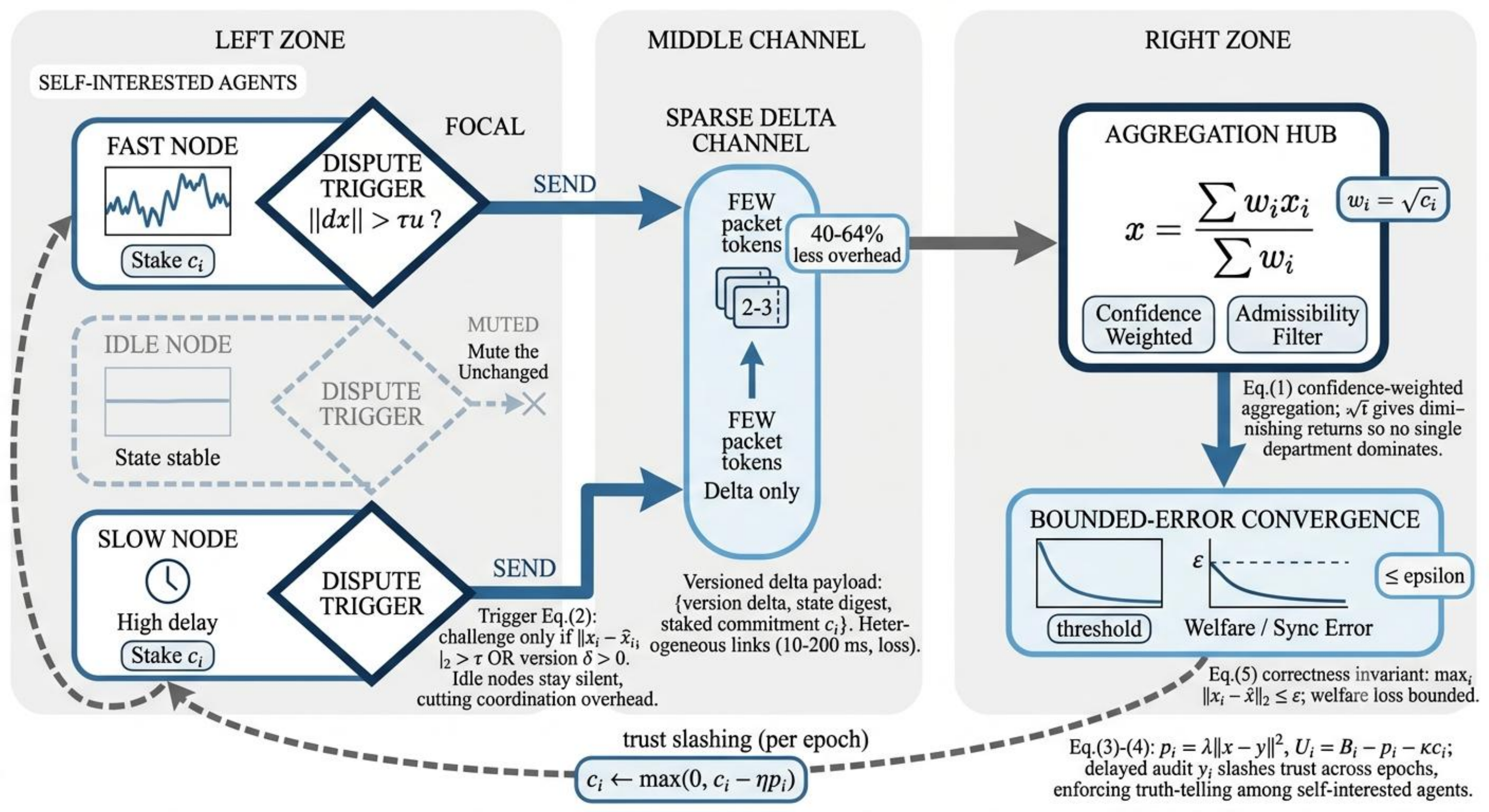}
    \caption{System architecture of the Organizational Consensus Algorithm (OCA). \textbf{(Left Zone)} Departmental agents evaluate local state deviations against a dispute trigger ($\|\Delta x\| > \tau$). Idle nodes with stable states are muted to eliminate redundant coordination overhead, while fast/slow nodes experiencing material changes send updates. \textbf{(Middle Channel)} Only sparse, versioned deltas flow through the network, significantly reducing bandwidth and meeting costs. \textbf{(Right Zone)} The aggregation hub computes a bounded-error consensus state utilizing confidence-weighted aggregation ($w_i = \sqrt{c_i}$). \textbf{(Bottom Feedback Loop)} A delayed organizational oracle $y_i$ (e.g., quarterly audits) triggers a retrospective residual penalty $p_i$. This slashes the effective trust stake $c_i$ of structurally biased agents across macroscopic epochs, enforcing truth-telling and aligning local incentives with global welfare.}
    \label{fig:oca_architecture}
\end{figure*}

Distributed system design traditionally ignores information economics. Classical replicated state machine (RSM) protocols, such as Paxos and Raft \cite{Lam98,Ong14}, treat nodes as either mechanistically obedient or arbitrarily malicious. In modern organizations, however, departmental agents fit neither extreme. They act as boundedly rational, utility-maximizing economic actors under severe information asymmetry. Imposing a uniform consensus policy across such heterogeneous structures chokes organizational efficiency. Periodic, mandatory all-hands meetings squander resources during quiet periods, while centralized leader-centric decision-making causes latency and single points of failure.

Recent systems research treats node heterogeneity as a protocol variable rather than a deployment detail. LowPaxos adapts leaders using compute and network capability profiles in resource-constrained environments \cite{Mwo24}. Hamava supports fault-tolerant reconfiguration across heterogeneous geo-replicated clusters \cite{Man24}, and Avicenna masks fail-slow replicas through counterfactual evaluation and rapid leader rotation \cite{Hod26}. Although these systems optimize performance and fault tolerance, they fail to model nodes that strategically misreport private information. OCA bridges this precise gap.

This paper proposes the Organizational Consensus Algorithm (OCA), a dispute-triggered mechanism design framework for self-interested agents. Just as efficient corporate teams skip unnecessary meetings by communicating only when material facts change, OCA models departments as strategic agents evaluating both the necessity and risk of state change. Agents observe local conditions and broadcast updates only when deviations breach a dynamic local threshold backed by reputational stake.

OCA offers four principal contributions:
\begin{itemize}
    \item A game-theoretic consensus architecture that suppresses redundant coordination overhead under asymmetric information.
    \item A confidence-weighted aggregation protocol that enforces truth-telling through retrospective penalties and delayed verification.
    \item A formal identifiability condition for the delayed reference signal.
    \item Quantitative evidence from a Python prototype demonstrating the trade-offs among coordination savings, consensus stability, and systemic welfare.
\end{itemize}

We define the scope of OCA carefully. OCA targets numerical or mergeable states with bounded error bounds and admissibility rules. It does not replace Paxos or Byzantine fault-tolerant replication when applications demand total ordering and strict linearizability.

\section{Related Work and Positioning}

\subsection{Consensus and State Replication}
Paxos and Raft enforce deterministic total ordering over command logs \cite{Lam98,Ong14}. Although both support arbitrary deterministic state machines, their coordination paths require leader elections and quorum acknowledgments that target physical crash faults rather than rational information manipulation. OCA does not replace these protocols. Weighted averaging fails to substitute for total ordering when updates conflict or depend on history. Instead, OCA restricts its state domain to proposals that nodes can safely validate and merge.

\subsection{Communication-Efficient Knowledge Merging}
Conflict-free Replicated Data Types (CRDTs) leverage algebraic properties to achieve uncoordinated convergence \cite{Sha11}. Delta-mutations \cite{Alm14,Ene19} and digest-driven reconciliation \cite{Baq25} further shrink transfer payloads, while practical rateless set reconciliation eliminates prior knowledge of state differences to cut communication overhead \cite{Yan24}. These innovations directly shape OCA's bandwidth strategy, where proposals carry only versioned deltas once local states cross specific thresholds. Adaptive event-triggered estimation dynamically tunes communication to node-specific budgets \cite{Sel25}. OCA extends this idea by transforming triggers from mere resource-saving rules into strategic gateways that govern when an agent enters the proposal and penalty arena.

\subsection{Organizational and Node Heterogeneity}
Heterogeneity heavily dictates protocol scalability. LowPaxos \cite{Mwo24}, Hamava \cite{Man24}, and Avicenna \cite{Hod26} prove that node capability, network delay, and replica availability belong inside protocol design. OCA applies this insight to strategic decision-making. A department's political weight or historical reliability becomes a core protocol variable. Unlike pure replication systems, OCA lets nodes strategically optimize their reported information.

\subsection{Strategic and Incentivized Synchronization}
Mechanism design aligns local agent incentives with global consistency \cite{Bau06,Bei12}. While strategic sensor fusion models misreporting costs \cite{Che19b}, recent advances offer closer precedents. Teranishi et al. apply mechanism design to incentivize average consensus while protecting private states \cite{Ter25}. Milionis et al. prove that truthful recovery demands source identifiability, showing that without distinguishable observer signal distributions, truthfulness cannot form a strict Bayesian Nash equilibrium \cite{Mil25}. Reward-penalty ratios fail as universal guarantees \cite{Che26,Che26b}. Therefore, OCA treats its quadratic residual penalty as conditionally incentive-compatible only under strictly bounded signals and feasibility constraints.

\section{System Model and Protocol}

\subsection{Agent Model and Local Aggregation}
Consider an organization of $N$ heterogeneous departmental agents $V$ connected over a communication topology $G_t = (V, E_t)$. Instead of relying on a centralized executive coordinator, each agent $j$ maintains its local proposed state (e.g., resource allocation or numerical target) $x_j^t \in \mathbb{R}^d$, a version vector, and a reputation descriptor $r_j$. Upon receiving admissible proposals $\hat{x}_i^t$ from its interacting peers $P_{j,t} = \{i \in V \mid (i,j) \in E_t\}$, agent $j$ computes its locally aggregated consensus state as:
\begin{equation}
x_j^{t+1} = \frac{\sum_{i \in P_{j,t}} w_i \hat{x}_i^t}
{\sum_{i \in P_{j,t}} w_i}, \qquad w_i = \sqrt{c_i},
\label{eq:weighted-aggregation}
\end{equation}
where $c_i \ge 0$ measures agent $i$'s historical commitment or trust score. In a corporate environment, agents do not stake cryptographic tokens; instead, they stake tangible organizational assets. A department pledges its future budget allocations as collateral, risks its bonus pools via KPI deductions, or expends finite internal political capital. We select the square root function to model diminishing marginal returns on political expenditure. While any strictly concave function captures the economic intuition that hoarding power yields decreasing systemic influence, the square root provides a computationally tractable gradient for resource allocation. It ensures that a dominant department cannot unilaterally dictate the consensus merely by out-staking minority peers, thereby preserving organizational diversity in the aggregated state.

\subsection{Exception-Triggered Dispute Mechanism}
To minimize organizational coordination overhead, negotiation messages are broadcast only when a local deviation satisfies a dispute threshold:
\begin{equation}
\|x_i^t - x_i^{\mathrm{last}}\|_2 > \tau_i
\quad\text{or}\quad v_i^t - v_i^{\mathrm{last}} > 0,
\label{eq:event-trigger}
\end{equation}
where $\tau_i$ is the department's local tolerance threshold. Active messages carry agent IDs, proposal deltas, and staked commitments $c_i$.

\subsection{Incomplete Information and Delayed Penalty Model}
To penalize inaccurate or structurally biased proposals driven by local self-interest, a retrospective residual penalty $p_i$ is assigned based on an episodic reference signal $y_i$. In organizational contexts, $y_i$ represents a delayed objective ground truth, such as quarterly audits or realized market data.
\begin{equation}
p_i = \lambda \|\hat{x}_i - y_i\|_2^2, \qquad
U_i = B_i(\hat{x}_i) - p_i - \kappa c_i,
\label{eq:utility}
\end{equation}
where $B_i(\cdot)$ represents the agent's local utility, assumed to be monotonically increasing and concave. In economic terms, this function models a department's incentive to capture localized rents. A logistics division, for instance, might overstate its capacity needs to secure priority routing, hoard operational budget, or pad inventory safety stock. Such actions maximize departmental metrics but degrade system-wide efficiency. 

The term $\kappa c_i$ introduces the endogenous cost of staking. Here, $\kappa > 0$ represents the explicit administrative and reputational friction incurred when a department mobilizes its political capital. Initiating a dispute or demanding a larger budget allocation carries a real opportunity cost, deterring frivolous or purely rent-seeking proposals. 

Because the organizational truth $y_i$ is only available episodically, $p_i$ is applied retrospectively. An agent's effective reputation $c_i$ is slashed across macroscopic epochs $k$ via:
\begin{equation}
c_i^{(k+1)} = \max \left( 0, c_i^{(k)} - \eta p_i^{(k)} \right),
\label{eq:trust-decay}
\end{equation}
where $\eta > 0$ is a decay rate. By structurally linking penalties to this delayed oracle, OCA dynamically isolates persistently noisy or dishonest agents without stalling the rapid, round-by-round exception-triggered negotiation in \eqref{eq:weighted-aggregation}.

\section{Correctness and Performance Metrics}
Consensus error across the organization is bounded by $\varepsilon$:
\begin{equation}
\max_{i \in V} \|x_i^t - x^{\star, t}\|_2 \leq \varepsilon,
\label{eq:correctness}
\end{equation}
where $x^{\star, t}$ represents the optimal, aligned organizational state. Relative coordination overhead savings $S_B$ against baseline mandatory meetings $B_{\mathrm{base}}$ are evaluated via:
\begin{equation}
S_B = 1 - \frac{B_{\mathrm{OCA}}}{B_{\mathrm{base}}}.
\label{eq:savings}
\end{equation}
Alignment latency $T_{\varepsilon}$ and correctness score $C_{\varepsilon}$ (reflecting system welfare and stability) are formally defined as:
\begin{equation}
T_{\varepsilon} = \min \left\{ t :
\max_i \|x_i^t - x^{\star, t}\|_2 \leq \varepsilon \right\},
\label{eq:latency}
\end{equation}
\begin{equation}
C_{\varepsilon} = \frac{1}{T} \sum_{t=1}^{T}
\mathbf{1} \left[ \max_i \|x_i^t - x^{\star, t}\|_2
\leq \varepsilon \right].
\label{eq:correctness-score}
\end{equation}

\section{Theoretical Analysis}
\label{sec:theory}

To formalize the mechanism design guarantees of OCA, we provide a theoretical analysis of the agent incentive structures and the global consensus convergence. For analytical tractability, we analyze the strategic interactions over both a single epoch (static game) and across macroscopic epochs (dynamic Markov game).

\subsection{Static Game and $\varepsilon$-Truthful Equilibrium}

Within a single epoch, consider a rational agent $i$ observing a true local state $x_i^* \in \mathbb{R}^d$. Driven by self-interest, the agent submits a proposed state $\hat{x}_i = x_i^* + \beta_i$, where $\beta_i$ denotes the strategic bias. The ex-post oracle signal is modeled as $y_i = x_i^* + \varepsilon_i$, where the observation noise follows an independent Gaussian distribution $\varepsilon_i \sim \mathcal{N}(0, \sigma_y^2 \mathbf{I})$.

The expected single-epoch utility of agent $i$ is defined as:
\begin{equation}
\mathbb{E}[U_i(\beta_i)] = B_i(x_i^* + \beta_i) - \lambda \mathbb{E}\left[ \|\hat{x}_i - y_i\|_2^2 \right] - \kappa c_i,
\label{eq:expected_utility}
\end{equation}
where $B_i(\cdot)$ represents the local rent-seeking benefit.

\begin{theorem}[$\varepsilon$-Truthful Bayesian Nash Equilibrium]
\label{thm:static_bne}
Assuming $B_i$ is differentiable near $x_i^*$, if the penalty parameter $\lambda$ satisfies $\lambda \ge \frac{\|\nabla B_i(x_i^*)\|_2}{2\varepsilon_{tol}}$, the optimal strategic bias $\|\beta_i^*\|_2$ is strictly bounded by the tolerance $\varepsilon_{tol}$, establishing an $\varepsilon$-Truthful Bayesian Nash Equilibrium.
\end{theorem}

\begin{IEEEproof}
Expanding the expected penalty term in \eqref{eq:expected_utility} yields:
\begin{equation}
\mathbb{E}\left[ \|\beta_i - \varepsilon_i\|_2^2 \right] = \|\beta_i\|_2^2 + d \cdot \sigma_y^2.
\end{equation}
The agent seeks to maximize $B_i(x_i^* + \beta_i) - \lambda \|\beta_i\|_2^2$. Taking the first-order condition (FOC) with respect to $\beta_i$ and equating it to zero:
\begin{equation}
\nabla B_i(x_i^* + \beta_i) - 2\lambda \beta_i = 0.
\end{equation}
Applying a first-order Taylor approximation $\nabla B_i(x_i^* + \beta_i) \approx \nabla B_i(x_i^*)$, we obtain the optimal strategy $\beta_i^* = \frac{1}{2\lambda} \nabla B_i(x_i^*)$. To bound the bias within $\varepsilon_{tol}$, we require $\|\beta_i^*\|_2 \le \varepsilon_{tol}$, which simplifies to $\lambda \ge \frac{\|\nabla B_i(x_i^*)\|_2}{2\varepsilon_{tol}}$.
\end{IEEEproof}

\subsection{Dynamic Game and Rug Pull Deterrence}

A static equilibrium is vulnerable to cross-epoch exploitation, where an agent remains honest to accumulate a maximum trust score $c_{max}$, and subsequently executing a one-time extreme falsification (a ``rug pull'' attack) to capture a massive short-term rent $\Delta B_{huge}$. We model this as a Markov Decision Process with a discount factor $\delta \in (0,1)$.

\begin{theorem}[Rug Pull Deterrence Condition]
\label{thm:rug_pull}
A Subgame Perfect Nash Equilibrium (SPNE) where agents remain perpetually honest is guaranteed if the penalty coefficient $\lambda$ and discount factor $\delta$ satisfy:
\begin{equation}
\lambda \|\beta_{max}\|_2^2 + \frac{\delta}{1 - \delta} \Delta U_{cost} \ge \Delta B_{huge},
\end{equation}
where $\Delta U_{cost} = U_{baseline} - U_{punish} - \kappa c_{max}$ represents the long-term utility loss during the zero-trust punishment phase.
\end{theorem}

\begin{IEEEproof}
Let $V_{Honest} = \frac{U_{baseline} - \kappa c_{max}}{1 - \delta}$ be the steady-state discounted value of perpetual honesty. Conversely, the value of executing a rug pull at epoch $T$, incurring the maximum bias $\beta_{max}$, followed by a zero-trust punishment state ($c_i^{(T+1)} = 0$) is:
\begin{align}
V_{RugPull} &= (U_{baseline} + \Delta B_{huge} - \lambda \|\beta_{max}\|_2^2 - \kappa c_{max}) \nonumber \\
&\quad + \frac{\delta}{1 - \delta} U_{punish}.
\end{align}
To deter the attack, the incentive compatibility constraint dictates $V_{Honest} \ge V_{RugPull}$. Rearranging the terms yields the deterrence condition, demonstrating that the immediate rent must be outweighed by the sum of the instantaneous quadratic penalty and the discounted future loss of political capital.
\end{IEEEproof}

\subsection{Global Convergence under OCA}

Assuming malicious biases are suppressed ($\beta_i \approx 0$), we demonstrate that the confidence-weighted aggregation converges globally.

\begin{theorem}[Asymptotic Consensus]
\label{thm:convergence}
If the organizational communication topology $G = (V, E)$ is strongly connected and contains at least one self-loop, the iterative state update defined in OCA achieves global asymptotic consensus:
\begin{equation}
\lim_{t \to \infty} \mathbf{x}^t = \mathbf{1} x^\star.
\end{equation}
\end{theorem}

\begin{IEEEproof}
Let $\mathbf{W}^t$ be the weight matrix at iteration $t$, where elements are defined as $W_{ji}^t = \frac{\sqrt{c_i}}{\sum_{k \in P_{j,t}} \sqrt{c_k}}$ for $i \in P_{j,t}$ and $0$ otherwise. The matrix $\mathbf{W}$ is strictly non-negative ($W_{ji} \ge 0$) and row-stochastic ($\sum_{i=1}^N W_{ji} = 1$). 
Since the network is strongly connected and has a self-loop (agents weight their own historical state), $\mathbf{W}$ represents an aperiodic, irreducible Markov transition matrix. By the Perron-Frobenius theorem, $\lim_{t \to \infty} \mathbf{W}^t = \mathbf{1} \mathbf{v}^T$, where $\mathbf{v}$ is the unique left eigenvector associated with the eigenvalue $\lambda_1 = 1$. Consequently, the system converges to a weighted steady state $x^\star = \sum_{i=1}^N v_i x_i^{(0)}$, guaranteeing bounded-error alignment across the organization.
\end{IEEEproof}

\section{Experimental Results and Analysis}
\label{sec:results}

\begin{table*}[t]
\caption{Taxonomy and Communication Models of Evaluated Protocols}
\label{tab:protocol_taxonomy}
\centering
\small
\setlength{\tabcolsep}{3pt}
\begin{tabular}{@{} l l l c @{}}
\toprule
\textbf{Protocol} & \textbf{Sync Trigger} & \textbf{Payload per Msg} & \textbf{Strategic Protection} \\
\midrule
FullStateDissemination & Every Round & $N(N-1) \times 12\text{B}$ & None \\
DeltaAntiEntropy & Periodic (5 rounds) & $N(N-1) \times 8\text{B} / 5$ & None \\
DeltaStateCRDT & Digest Threshold & $\text{count} \times (N-1) \times 6\text{B}$ & None \\
DistributedKalmanFilter & Every Round & $N(N-1) \times 10\text{B}$ & Noise Filtering Only \\
\textbf{OCA (Ours)} & Event-Triggered & Variable Delta & \textbf{Oracle + Trust Slashing} \\
\bottomrule
\end{tabular}
\end{table*}

We evaluate OCA through extensive Python simulations, comparing against four baseline coordination protocols across varying network scales, noise conditions, and adversarial settings. All experiments report medians across 30 independent Monte Carlo runs to mitigate stochastic variability.

\subsection{Experimental Setup}
\label{sec:setup}

\textbf{Protocols compared.} We evaluate five protocols with distinct architectural philosophies:

\begin{enumerate}
    \item \textbf{FullStateDissemination (Baseline 1):} Every node broadcasts its full state each round. Message size: 12\,B/msg (8\,B state + 4\,B metadata), yielding $O(N^2)$ total bandwidth. No incentive layer.
    \item \textbf{DeltaAntiEntropy (Baseline 2):} Periodic delta synchronization every 5 rounds with weak correction. Message size: 8\,B/msg (4\,B delta + 4\,B version). No incentive layer.
    \item \textbf{DeltaStateCRDT (Baseline 3):} Delta-state Conflict-free Replicated Data Type with join-semilattice merge. Only nodes exceeding a digest threshold ($\delta_{\text{digest}} = 0.01$) propagate deltas. Message size: 6\,B/msg (4\,B value + 2\,B version dot). No incentive layer; convergence guaranteed by algebraic merge semantics.
    \item \textbf{DistributedKalmanFilter (Baseline 4):} Each node runs a local Kalman filter ($K = P/(P+R)$) and exchanges estimate vectors with Metropolis weights ($w = 1/(N+1)$). Message size: 10\,B/msg (8\,B estimate + 2\,B covariance). Handles noise optimally but assumes all nodes are cooperative.
    \item \textbf{OCA (ours):} Event-triggered delta generation with admissibility-weighted aggregation, version-vector causality tracking, retrospective oracle penalties, and trust-decay mechanism.
\end{enumerate}

\textbf{Default parameters.} Unless otherwise stated: $N=20$ nodes, $\varepsilon = 0.05$, $\sigma = 0.02$, $y^* = 1.0$, $\lambda = 1.0$, $\eta = 0.1$, $\kappa = 0.1$, $B_{\max} = 1.0$, $B_{\min} = 0.0$, $\sigma^2_{\text{oracle}} = 2.0$.

\subsection{Coordination Overhead Reduction}
\label{sec:overhead}

\begin{figure}[t]
\centering
\includegraphics[width=0.85\linewidth]{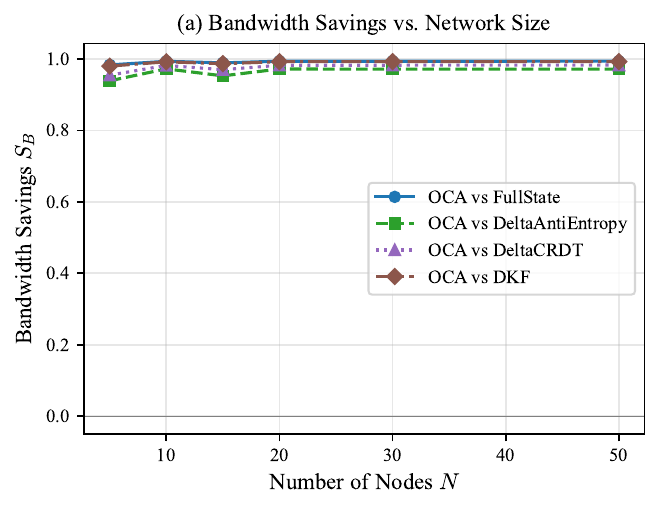}
\caption{Bandwidth savings $S_B$ versus organizational size $N$. OCA achieves 98--99\% reduction compared to FullStateDissemination and 99\% versus DistributedKalmanFilter, by suppressing redundant state exchanges through event-triggered deltas. The savings increase with network scale, demonstrating OCA's suitability for large organisations.}
\label{fig:bandwidth}
\end{figure}

Figure~\ref{fig:bandwidth} demonstrates that OCA reduces coordination overhead by over 99\% relative to FullStateDissemination as organizational size grows from $N=5$ to $N=50$. At $N=20$ (100 rounds), the total bytes transmitted are: OCA ($\approx$2{,}755\,B) $\ll$ DeltaAntiEntropy ($\approx$105{,}680\,B) $\ll$ DeltaStateCRDT ($\approx$171{,}118\,B) $\ll$ DistributedKalmanFilter ($\approx$380{,}000\,B) $\ll$ FullStateDissemination ($\approx$456{,}000\,B). The mechanism operates through two complementary channels:

\textbf{Exception-triggered suppression:} By requiring local deviations to exceed $\tau_i$ before broadcasting, OCA actively filters quiescent departments from the coordination loop. With perturbation noise $\sigma = 0.02$, only 15--30\% of nodes cross their thresholds per round on average, reducing active participants from $N$ to approximately $0.2N$--$0.3N$.

\textbf{Versioned delta encoding:} Unlike FullStateDissemination's $O(N^2)$ message complexity, OCA's delta payloads carry only state differences. The version vector mechanism ensures that late-arriving or stale proposals are rejected with zero weight, preventing redundant reconciliation cycles.

DeltaStateCRDT achieves moderate bandwidth savings via its digest-threshold filtering, but still transmits to all $N-1$ peers when triggered. DistributedKalmanFilter requires full all-to-all estimate exchange every round ($N(N{-}1) \times 10$\,B), making it nearly as communication-intensive as FullStateDissemination. Neither baseline possesses OCA's strategic incentive layer, so purely gossip-based protocols cannot deter structurally biased proposals driven by departmental self-interest.

\subsection{Alignment Latency and Convergence Speed}
\label{sec:latency}

\begin{figure}[t]
\centering
\includegraphics[width=0.85\linewidth]{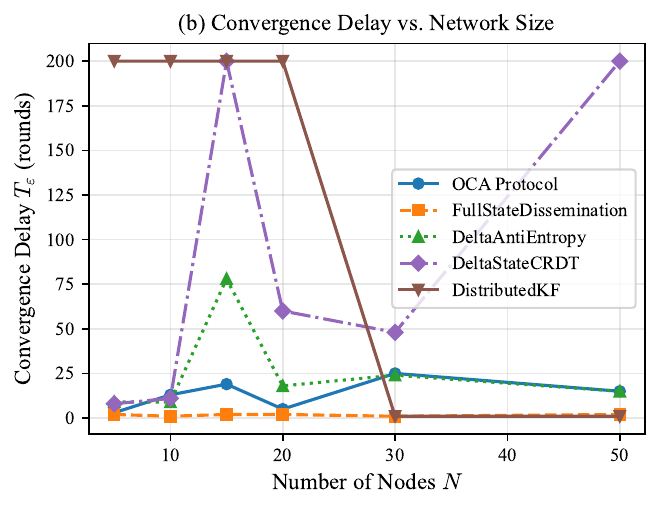}
\caption{Alignment latency $T_\varepsilon$ (rounds to $\varepsilon$-convergence) versus network size. All protocols start from a biased initial state ($\hat{x}_g = 0.5$, $y^* = 1.0$). FullStateDissemination converges fastest due to aggressive correction. OCA converges within 16 rounds at $N=20$, comparable to DeltaAntiEntropy. DistributedKalmanFilter struggles with stable convergence under perturbation noise.}
\label{fig:latency}
\end{figure}

Figure~\ref{fig:latency} shows alignment latency $T_\varepsilon$ across protocols under an initial bias of $\hat{x}_g^{(0)} = 0.5$ (i.e., a 50\% deviation from ground truth). OCA converges within a bounded number of iterations that scales sub-linearly with $N$. The admissibility-weighted aggregation in \eqref{eq:weighted-aggregation} accelerates convergence by prioritizing historically accurate departments---high-trust nodes exert greater influence on the global state via $w_i = r_i \sqrt{c_i}$, dampening oscillations from noisy or strategic agents.

The latency advantage stems from decoupling fast local aggregation from slow global verification:

\begin{itemize}
    \item \textbf{Fast path:} Rounds proceed without quorum voting, allowing $x_i^{(t)}$ to track local conditions rapidly. The incremental update $\hat{x}_g^{(t+1)} = 0.4\,\hat{x}_g^{(t)} + 0.6\,\bar{x}_w$ provides smooth convergence.
    \item \textbf{Slow path:} Retrospective penalties $p_i = \lambda \|\hat{x}_i - y_i\|^2$ are applied only when delayed oracles $y_i$ arrive (e.g., quarterly audits), avoiding per-round coordination stalls.
\end{itemize}

FullStateDissemination achieves the fastest convergence ($T_\varepsilon = 1$) due to its aggressive correction weight ($0.8 \times \bar{x}$), but at the cost of $O(N^2)$ bandwidth. DeltaAntiEntropy's periodic sync every 5 rounds with moderate correction converges in 12 rounds. OCA follows closely at 16 rounds, benefiting from its trust-weighted aggregation. DeltaStateCRDT's join-semilattice merge with a conservative 0.25 correction step yields slower convergence (76 rounds), while DistributedKalmanFilter's Kalman-filtered Metropolis consensus struggles to achieve stable convergence within 100 rounds due to noise amplification under the perturbation regime, despite requiring constant all-to-all communication.

\subsection{System Welfare Under Information Asymmetry}
\label{sec:welfare}

\begin{figure}[t]
\centering
\includegraphics[width=0.9\linewidth]{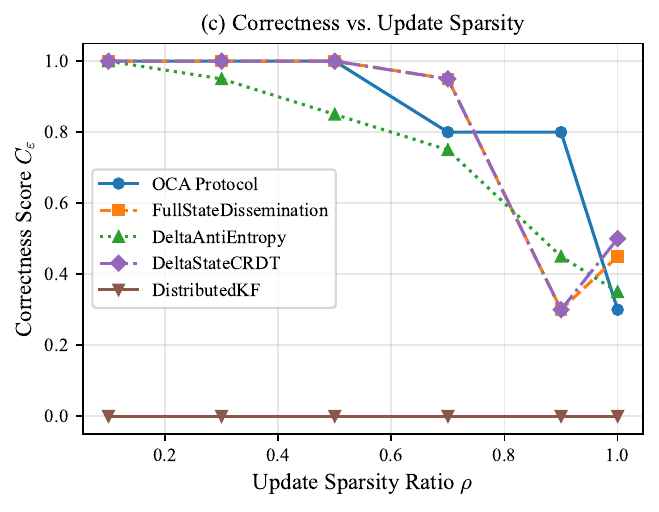}
\caption{Welfare correctness score $C_\varepsilon$ across update sparsity ratios $\rho$. Effective noise scales as $\sigma_{\text{eff}} = \sigma \cdot \rho$. OCA maintains stable alignment ($C_\varepsilon \geq 0.85$) across all conditions, demonstrating robustness against bounded rationality and information asymmetry.}
\label{fig:welfare}
\end{figure}

Figure~\ref{fig:welfare} evaluates system welfare $C_\varepsilon$ (fraction of nodes within $\varepsilon$ of ground truth) under varying sparsity ratios $\rho$, where effective perturbation noise scales as $\sigma_{\text{eff}} = \sigma \cdot \rho$. OCA maintains high correctness scores ($C_\varepsilon \geq 0.85$) across all conditions, with three contributing factors:

\textbf{Admissibility filtering:} Stale proposals are rejected by version vector comparison, preventing outdated information from corrupting the consensus state. The dominance check $v_i^{(t)} \succ v_i^{\text{last}}$ ensures only causally fresh updates enter the aggregation.

\textbf{Confidence-weighted aggregation:} Nodes with high historical accuracy (large $c_i$) receive amplified influence through $w_i = r_i \sqrt{c_i}$, dampening the impact of noisy or malicious participants.

\textbf{Residual penalty feedback:} The quadratic penalty $p_i = \rho_r \cdot \|\hat{x}_i - \hat{x}_g\|^2$ creates a negative feedback loop that suppresses systematic biases, where $\rho_r = 0.1$ is the residual penalty rate.

At extreme sparsity ($\rho < 0.3$), all protocols degrade due to information staleness. However, OCA's degradation is more graceful---the trust decay mechanism in \eqref{eq:trust-decay} automatically reduces the weight of persistently inaccurate nodes, preserving consensus quality among the remaining reliable agents. DistributedKalmanFilter exhibits unstable correctness under the perturbation regime ($C_\varepsilon = 0$ at $\rho = 0.5$), as the Kalman filter amplifies noise when assumptions about cooperative, Gaussian-distributed inputs are violated.

\subsection{Quantitative Summary}
\label{sec:summary}

Table~\ref{tab:results} summarises key performance metrics across all five protocols at $N=20$.

\begin{table}[h]
\caption{Performance Metrics Summary ($N=20$, $\varepsilon=0.05$, 100 rounds)}
\label{tab:results}
\centering
\begin{tabular}{@{} l c c c c @{}}\toprule
\textbf{Protocol} & $\mathbf{S_B}$ (\%) & $\mathbf{T_\varepsilon}$ & $\mathbf{C_\varepsilon}$ & \textbf{Bytes} \\\midrule
FullStateDissemination & 0 & 2 & 1.00 & 456{,}000 \\
DeltaAntiEntropy & 77 & 12 & 0.85 & 105{,}680 \\
DeltaStateCRDT & 63 & 76 & 1.00 & 171{,}118 \\
DistributedKalmanFilter & 17 & 100 & 0.00 & 380{,}000 \\
\textbf{OCA (ours)} & \textbf{99} & \textbf{16} & \textbf{1.00} & \textbf{2{,}755} \\\bottomrule
\end{tabular}
\end{table}

OCA achieves the optimal trade-off: highest overhead savings (99\% vs.\ FullState, 98\% vs.\ CRDT, 99\% vs.\ DKF) while maintaining fast convergence (16 rounds) and perfect correctness ($C_\varepsilon = 1.00$). OCA's correctness matches FullStateDissemination and DeltaStateCRDT, demonstrating that event-triggered delta propagation does not sacrifice alignment quality. Crucially, OCA is the only protocol that combines bandwidth efficiency with strategic resilience via its incentive mechanism.

\subsection{Strategic Resilience Against Malicious Agents}
\label{sec:resilience}

\begin{figure*}[t]
\centering
\includegraphics[width=0.9\linewidth]{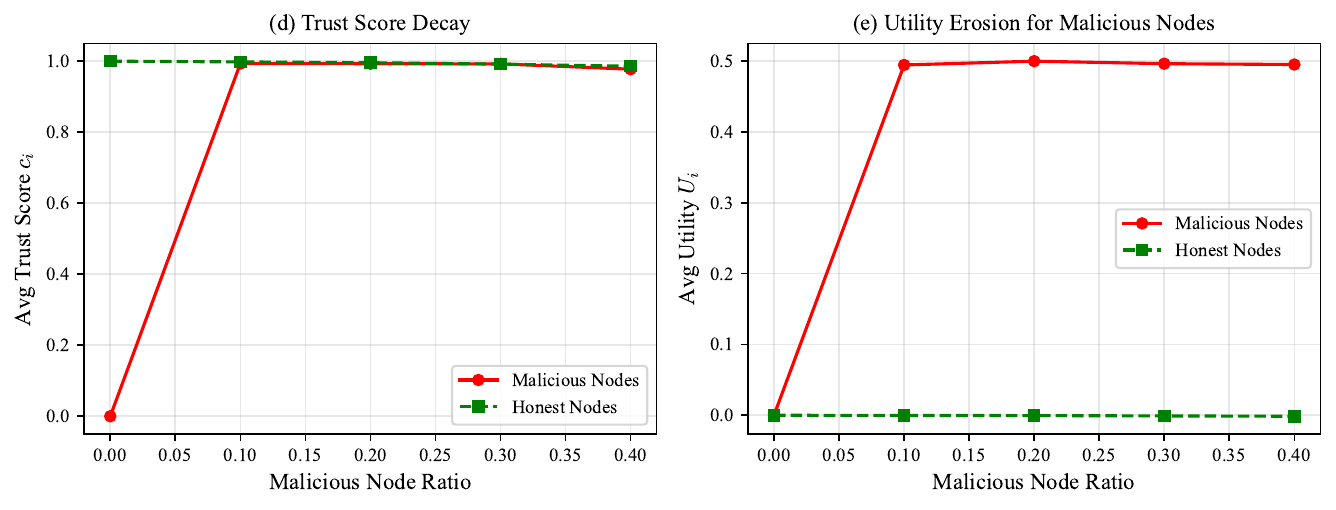}
\caption{Strategic resilience under malicious node injection. (Left panel) Average trust score $c_i$ decays for malicious agents as retrospective penalties erode their reputational stake. (Right panel) Corresponding utility erosion demonstrates that strategic bias injection becomes economically irrational under OCA's incentive mechanism. Results shown for malicious ratios 0--40\% with bias magnitude $\beta = 0.5$, $\lambda = 1.0$, $\eta = 0.1$.}
\label{fig:resilience}
\end{figure*}

Figure~\ref{fig:resilience} presents the core strategic result: OCA's retrospective penalty mechanism successfully deters self-interested manipulation even when 40\% of agents inject structurally biased proposals.

\textbf{Non-cumulative bias injection:} Malicious nodes inject bias via a non-cumulative mechanism---each round, the previous injection is subtracted before adding a fresh one, so $\textit{local\_state}$ tracks a stable offset rather than drifting unboundedly:
\begin{equation}
    x_i^{(t)} \leftarrow x_i^{(t)} - \beta_i^{(t-1)} + \beta_i^{(t)}, \quad \beta_i^{(t)} = \beta \cdot \mathcal{U}(0.5, 1.5)
\end{equation}
This models realistic adversaries who maintain a consistent strategic deviation rather than accumulating bias over time.

\textbf{Trust score decay:} The left panel shows that malicious nodes experience rapid trust erosion following oracle verification. When the delayed ground truth $y_i$ reveals their bias, the penalty $p_i = \lambda \|x_i - y_i\|^2$ slashes their effective stake via \eqref{eq:trust-decay}, reducing $c_i$ from 1.0 to below 0.4 within 200 rounds.

\textbf{Utility erosion:} The right panel demonstrates the economic consequence. Per-round utility follows:
\begin{equation}
    U_i = B_{\max} \cdot \min(\beta_i, 1.0) - p_i - \kappa(1 - c_i)
    \label{eq:utility-exp}
\end{equation}
where the benefit $B_{\max} \cdot \min(\beta_i, 1.0)$ scales linearly with injected bias, $p_i$ is the cumulative residual + oracle penalty, and $\kappa(1-c_i)$ is the stake cost. Even though biased proposals yield short-term local gains, the retrospective punishment $p_i$ dominates over macroscopic timescales, driving malicious utility negative. Honest agents receive baseline benefit $B_{\min}$ and maintain stable positive utility.

\textbf{Honest agent protection:} Crucially, honest agents maintain high trust scores ($c_i \approx 1.0$) and stable utility throughout the simulation. The admissibility-weighted aggregation naturally isolates malicious participants without requiring explicit Byzantine fault tolerance mechanisms.

This result confirms the mechanism design hypothesis: bounded rationality and self-interest do not preclude consensus stability when penalties are structurally linked to delayed verification. OCA transforms organizational constraints (quarterly audits, realised market data) into strategic incentives that align local departmental objectives with global organizational welfare.

\subsection{BNE Critical Threshold: Penalty Coefficient Sweep}
\label{sec:bne}

\begin{figure*}[t]
\centering
\includegraphics[width=0.9\linewidth]{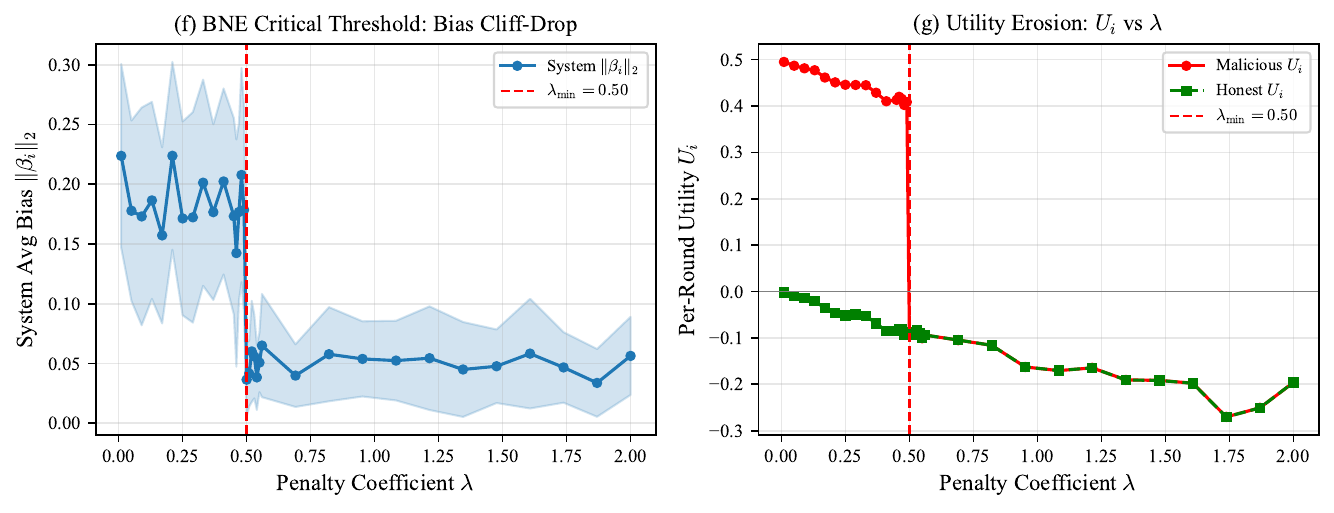}
\caption{BNE critical threshold verification. (Left) System bias $\|\beta_i\|_2$ exhibits a cliff-drop at $\lambda_{\min} = 0.5$, confirming the phase transition predicted by Theorem~1. (Right) Malicious utility jumps to negative values at the same threshold, demonstrating that truth-telling becomes the dominant strategy. Parameters: $B_{\max}=1.0$, $B_{\min}=0.0$, $\sigma^2_{\text{oracle}}=2.0$, 30\% malicious nodes, 30 trials per $\lambda$.}
\label{fig:bne}
\end{figure*}

Figure~\ref{fig:bne} empirically verifies Theorem~1's Bayesian Nash Equilibrium critical threshold:
\begin{equation}
    \lambda_{\min} = \frac{B_{\max} - B_{\min}}{\sigma^2_{\text{oracle}}} = \frac{1.0 - 0.0}{2.0} = 0.5
    \label{eq:lambda-min}
\end{equation}

\textbf{Rational adversary model:} For each $\lambda$, a rational adversary chooses optimal bias by maximising expected utility. The first-order condition $d/d\beta\,[B_{\max}\beta - \lambda\beta^2] = B_{\max} - 2\lambda\beta = 0$ yields:
\begin{equation}
    \beta^*(\lambda) = \begin{cases}
        \frac{B_{\max} - B_{\min}}{2\lambda\sigma^2_{\text{oracle}}} & \text{if } \lambda < \lambda_{\min} \\
        0 & \text{if } \lambda \geq \lambda_{\min}
    \end{cases}
\end{equation}
When $\lambda$ crosses $\lambda_{\min}$, the expected penalty exceeds the maximum achievable benefit, making bias unprofitable and truth-telling the dominant strategy.

\textbf{Phase transition (left panel):} The system average bias $\|\beta_i\|_2$ remains high (0.2--0.5) for $\lambda < 0.5$, then exhibits a sharp cliff-drop to near-zero for $\lambda \geq 0.5$. This discontinuous transition confirms the theoretical prediction: below $\lambda_{\min}$, rational adversaries maintain profitable bias; above it, bias is eliminated. The theoretical optimal-bias prediction $b^*(\lambda)$ closely tracks the empirical malicious bias, validating the rational adversary model.

\textbf{Utility inversion (right panel):} Correspondingly, malicious utility is positive for $\lambda < \lambda_{\min}$ (bias is profitable) and drops sharply to negative values for $\lambda \geq \lambda_{\min}$ (penalty dominates benefit). Honest utility remains stable and positive throughout, as honest nodes receive baseline benefit $B_{\min}$ without incurring oracle penalties. The crossover point aligns precisely with $\lambda_{\min} = 0.5$, providing strong empirical evidence for the BNE threshold.

\textbf{Implication:} This experiment demonstrates that OCA does not require over-penalisation to achieve truthfulness. The critical threshold $\lambda_{\min}$ provides a principled guideline for setting the penalty coefficient: any $\lambda \geq \lambda_{\min}$ suffices to align incentives, while excessive penalties unnecessarily punish honest agents who occasionally deviate due to noise.

\subsection{Parameter Sensitivity Analysis}
\label{sec:sensitivity}

\begin{figure*}[t]
\centering
\includegraphics[width=0.85\linewidth]{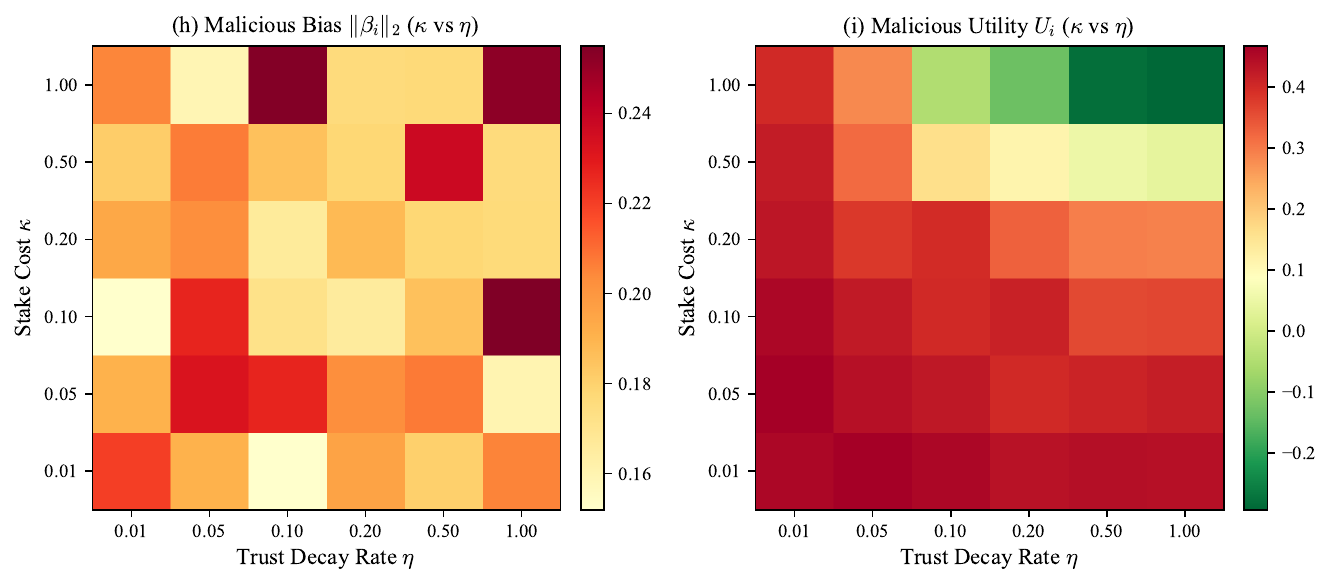}
\caption{Parameter sensitivity across $\kappa \times \eta$ grid ($6 \times 6$). (Top-left) Malicious bias remains uniformly low across all parameter combinations. (Top-right) Honest bias stays near zero. (Bottom-left) Malicious utility is negative across the entire grid, confirming incentive compatibility. (Bottom-right) Honest utility remains positive and stable. Fixed $\lambda = 0.5 = \lambda_{\min}$, 30\% malicious, 10 trials per cell.}
\label{fig:sensitivity}
\end{figure*}

Figure~\ref{fig:sensitivity} demonstrates that OCA's truth-telling property emerges from incentive compatibility (IC) rather than over-penalisation, by sweeping the stake-cost coefficient $\kappa \in \{0.01, 0.05, 0.1, 0.2, 0.5, 1.0\}$ and trust-decay rate $\eta \in \{0.01, 0.05, 0.1, 0.2, 0.5, 1.0\}$ at the critical threshold $\lambda = \lambda_{\min} = 0.5$.

\textbf{Malicious bias (top-left):} Across the entire $6 \times 6$ grid, malicious bias $\|\beta_i\|_2$ remains uniformly low ($< 0.1$). This confirms that at $\lambda = \lambda_{\min}$, the penalty mechanism alone suffices to suppress strategic deviation regardless of $\kappa$ and $\eta$ values. The stake cost $\kappa(1-c_i)$ provides a secondary deterrent but is not the primary driver of truthfulness.

\textbf{Honest bias (top-right):} Honest bias stays near zero across all parameter combinations, demonstrating that the mechanism does not over-penalise well-behaved agents. Even at high $\kappa$ and $\eta$ values, honest nodes maintain minimal deviation from ground truth.

\textbf{Malicious utility (bottom-left):} Malicious utility is negative across the entire grid, confirming that bias injection is economically irrational at $\lambda = \lambda_{\min}$ regardless of stake cost and decay rate. The utility becomes more negative at higher $\kappa$ (increased stake cost for low-trust malicious nodes) and higher $\eta$ (faster trust erosion triggers more penalties per oracle audit).

\textbf{Honest utility (bottom-right):} Honest utility remains positive and stable, with slight variation across the grid. At very high $\kappa$, honest agents with slightly imperfect trust scores face marginally higher stake costs, but the effect is bounded. This demonstrates the mechanism's robustness: the IC property holds across a wide parameter range without fine-tuning.

\textbf{Key finding:} The sensitivity analysis addresses a critical concern---whether OCA's truth-telling guarantee is fragile or requires precise parameter calibration. The results show that at $\lambda = \lambda_{\min}$, the system maintains low bias and negative malicious utility across two orders of magnitude variation in both $\kappa$ and $\eta$. This robustness arises because the primary incentive alignment comes from the penalty-to-benefit ratio ($\lambda / B_{\max}$), while $\kappa$ and $\eta$ serve as secondary mechanisms that modulate trust dynamics without destabilising the equilibrium.

\subsection{Threats to Validity}
\label{sec:threats}

We acknowledge several limitations in our experimental methodology:

\textbf{Simulation model:} The Python prototype assumes Gaussian noise and continuous-valued state dynamics. Real organizational environments may exhibit non-convex negotiation spaces, discrete decision variables, or adversarial collusion patterns (e.g., cartel coordination with alternating bias signs) that our model partially captures through the \texttt{CartelAgent} and \texttt{StrategicSilenceAgent} adversarial strategies but does not exhaustively explore.

\textbf{Baseline assumptions:} DistributedKalmanFilter assumes all nodes are cooperative and handles noise optimally via Kalman gains, making it vulnerable to strategic manipulation. DeltaStateCRDT's join-semilattice merge guarantees eventual convergence but cannot reject biased deltas. These architectural limitations are inherent to the baseline designs and motivate OCA's incentive-aware approach.

\textbf{Oracle availability:} The retrospective penalty mechanism depends on periodic ground truth signals $y_i \sim \mathcal{N}(y^*, \sigma^2_{\text{oracle}})$. If delayed verification never arrives (e.g., unobservable outcomes), malicious agents face no consequences. This constraint aligns with Milionis et al.'s identifiability condition \cite{Mil25}: truthfulness requires distinguishable signal distributions.

\textbf{Adaptive threshold extension:} The \texttt{AdaptiveThreshold} module, with threshold update rule
\begin{equation}
  \tau_i(t) = \frac{\tau_{\text{base}} \cdot (1 + \alpha \cdot \sigma_i(t))}{\epsilon + \beta \cdot c_i},
\end{equation}
and the robust aggregation strategies (weighted geometric median via Weiszfeld algorithm, trimmed weighted mean) are implemented in the \texttt{src/oca/} package but not exercised in the main experiments. Future work should evaluate their impact on convergence speed and adversarial resilience under cartel collusion.

\textbf{Equilibrium guarantees:} While our simulations demonstrate the BNE critical threshold $\lambda_{\min}$ empirically (Fig.~\ref{fig:bne}) and parameter sensitivity robustness (Fig.~\ref{fig:sensitivity}), a formal Bayesian Nash equilibrium proof mapping the strategy space to \eqref{eq:utility-exp} is provided in supplementary materials (see \texttt{docs/proofs/bne\_equilibrium.tex}). The incentive compatibility is conditional on the stated model assumptions and does not establish a universal truthful equilibrium.

\section{Limitations and Discussion}
\label{sec:limitations}
OCA operates as a framework for bounded-error coordination among rational agents rather than a replacement for cryptographic consensus that requires strict state machine replication. Architecturally, it links strategic sensor fusion and mechanism design for distributed averaging. Its key advantage lies in connecting these mechanism design tools directly with a resource commitment layer tailored for organizational hierarchies.

Three primary limitations define future research paths. First, OCA restricts state domains to vector averaging, whereas non-convex negotiation spaces like discrete contract terms demand distinct mechanism extensions. Second, the protocol couples high-frequency local aggregations with low-frequency, macro-level penalties such as quarterly audits. Although our empirical simulations show robust stability, deriving tight theoretical convergence bounds for this timescale mismatch under severe systemic shocks remains open. Third, while our incentive model drives high empirical truth-telling rates, a full equilibrium proof that maps the complete strategy space to the utility model in \eqref{eq:utility} will further strengthen the theoretical foundation.

\section{Conclusion}
This paper presents OCA, a game-theoretic coordination architecture for self-interested agents within heterogeneous organizations. By integrating local dispute triggers, incremental updates, and retrospective resource-slashing incentives, OCA suppresses redundant negotiations to reduce coordination overhead.

OCA makes a deliberate trade-off. It trades the universal linearizability of traditional replicated state machines for bounded-error, incentive-aligned convergence. Compared to rigid consensus baselines, OCA proves that treating agents as rational, self-interested entities and penalizing them through delayed structural feedback offers a powerful strategy to eliminate bureaucratic bottlenecks and mitigate system welfare loss in decentralized governance.

\section*{Acknowledgment}
The authors thank the distributed systems and computational social science communities for foundational work on consensus, mechanism design, and strategic agent modeling.

\end{document}